\documentclass[final,onefignum,onetabnum]{siamart171218}
\usepackage[utf8]{inputenc}
\usepackage[T1]{fontenc} % 8-bit encoding
\usepackage{color}
\usepackage[leqno]{amsmath}
\usepackage{amsfonts, amssymb, commath} % Math packages
\usepackage{mathtools}
\usepackage{graphicx}
\usepackage{mathrsfs}
\usepackage{caption}
\usepackage{listings}
\usepackage[]{eufrak}
\usepackage[]{hyperref}
\usepackage[]{float}
\usepackage{tikz}
\usetikzlibrary{calc, arrows,matrix,positioning,scopes}
\usetikzlibrary{shapes.geometric}
\usetikzlibrary{decorations.markings}
\usepackage{tikz-cd}
\usetikzlibrary{cd}
\usepackage{booktabs}
\usepackage{enumerate}
\usepackage{eufrak}
\usepackage{MnSymbol}
\usepackage{fancyhdr}
\usepackage{tikz}
\usepackage{algorithm}
\usepackage{algorithmic}
\usetikzlibrary{matrix,decorations.pathreplacing,calc}
\pgfkeys{tikz/mymatrixenv/.style={decoration=brace,every left delimiter/.style={xshift=3pt},every right delimiter/.style={xshift=-3pt}}}
\pgfkeys{tikz/mymatrix/.style={matrix of math nodes,left delimiter={(}, right delimiter={)}, inner sep=1pt,column sep=1.5em,row sep=0.5em,nodes={inner sep=0pt}}}
\pgfkeys{tikz/mymatrixbrace/.style={decorate,thick}}
\newcommand\mymatrixbraceoffseth{0.5em}
\newcommand\mymatrixbraceoffsetv{0.2em}

\newcommand*\mymatrixbraceright[4][m]{
    \draw[mymatrixbrace] ($(#1.north west)!(#1-#3-1.south west)!(#1.south west)-(\mymatrixbraceoffseth,0)$)
        -- node[left=2pt] {#4} 
        ($(#1.north west)!(#1-#2-1.north west)!(#1.south west)-(\mymatrixbraceoffseth,0)$);
}

\newcommand*\mymatrixbracetop[4][m]{
    \draw[mymatrixbrace] ($(#1.north west)!(#1-1-#2.north west)!(#1.north east)+(0,\mymatrixbraceoffsetv)$)
        -- node[above=2pt] {#4} 
        ($(#1.north west)!(#1-1-#3.north east)!(#1.north east)+(0,\mymatrixbraceoffsetv)$);
}

\usepackage{cuted}
\usepackage{braket,amsfonts}

\usepackage{array}

\usepackage[caption=false]{subfig}
\usepackage{pgfplots}

\newsiamthm{claim}{Claim}
\newsiamremark{remark}{Remark}
\newsiamremark{hypothesis}{Hypothesis}
\crefname{hypothesis}{Hypothesis}{Hypotheses}

\usepackage{algorithmic}

\usepackage{graphicx,epstopdf}

\Crefname{ALC@unique}{Line}{Lines}

\usepackage{amsopn}

\usepackage{xspace}
\usepackage{bold-extra}
\usepackage[most]{tcolorbox}

\colorlet{texcscolor}{blue!50!black}
\colorlet{texemcolor}{red!70!black}
\colorlet{texpreamble}{red!70!black}
\colorlet{codebackground}{black!25!white!25}

\lstdefinestyle{siamlatex}{%
  style=tcblatex,
  texcsstyle=*\color{texcscolor},
  texcsstyle=[2]\color{texemcolor},
  keywordstyle=[2]\color{texemcolor},
  moretexcs={cref,Cref,maketitle,mathcal,text,headers,email,url},
}

\tcbset{%
  colframe=black!75!white!75,
  coltitle=white,
  colback=codebackground, % bottom/left side
  colbacklower=white, % top/right side
  fonttitle=\bfseries,
  arc=0pt,outer arc=0pt,
  top=1pt,bottom=1pt,left=1mm,right=1mm,middle=1mm,boxsep=1mm,
  leftrule=0.3mm,rightrule=0.3mm,toprule=0.3mm,bottomrule=0.3mm,
  listing options={style=siamlatex}
}

\newtcblisting[use counter=example]{example}[2][]{%
  title={Example~\thetcbcounter: #2},#1}

\newtcbinputlisting[use counter=example]{\examplefile}[3][]{%
  title={Example~\thetcbcounter: #2},listing file={#3},#1}

\DeclareTotalTCBox{\code}{ v O{} }
{ %fontupper=\ttfamily\color{texemcolor},
  fontupper=\ttfamily\color{black},
  nobeforeafter,
  tcbox raise base,
  colback=codebackground,colframe=white,
  top=0pt,bottom=0pt,left=0mm,right=0mm,
  leftrule=0pt,rightrule=0pt,toprule=0mm,bottomrule=0mm,
  boxsep=0.5mm,
  #2}{#1}

\patchcmd\newpage{\vfil}{}{}{}
\newcommand{\ca}[1]{\mathcal{#1}}

\newcommand{\bb}[1]{\mathbb{#1}}

\newcommand{\rank}{\text{rank}}

\makeatletter
\newcommand{\leqnomode}{\tagsleft@true}
\newcommand{\reqnomode}{\tagsleft@false}
\makeatother
\renewcommand{\theenumi}{\alph{enumi}}

\let\lra\Leftrightarrow
\DeclarePairedDelimiter\floor{\lfloor}{\rfloor}
 \newtheorem{observation}[theorem]{Observation}
\newtheorem{example1}{Example}

\begin{tcbverbatimwrite}{tmp_\jobname_header.tex}
\title{On Topological and Metrical Properties of
Stabilizing Feedback Gains:
the MIMO Case\thanks{Technical Report; Robotics, Aerospace, and Information Networks Laboratory, Department of Aeronautics and Astronautics, University of Washington, April 4, 2019}}
\author{Jingjing Bu\thanks{Department of Electrical \& Computer Engineering, University of Washington, Seattle, WA 98195}
\and Afshin Mesbahi\thanks{Department of Aeronautics and Astronautics, University of Washington, Seattle, WA 98195}
\and Mehran Mesbahi\thanks{Department of Aeronautics and Astronautics, University of Washington, Seattle, WA 98195}}

\headers{Topological Properties of Stabilizing Feedback Gains: MIMO Case}{Jingjing Bu, Afshin Mesbahi, and Mehran Mesbahi}
\end{tcbverbatimwrite}
\input{tmp_\jobname_header.tex}

\ifpdf
\hypersetup{ pdftitle={MIMO_TOPO.tex} }
\fi

\begin{document}
\maketitle

%% ------------------------------------------------------------------
%% ABSTRACT
%% ------------------------------------------------------------------
\begin{tcbverbatimwrite}{tmp_\jobname_abstract.tex}
\begin{abstract}
  %% another version of abstract
  % In this paper, we present a fairly complete results on various topological and metrical analysis of feedback stabilization for multiple-input-multiple-output (MIMO) continuous and discrete time linear-time-invariant (LTI) systems. In addition, we show for the set of structured stabilizing state-feedback gains is in general disconnected and can have exponentially many connected components. Sufficient conditions are derived to ensure the connectedness of the set of structured stabilizing feedback gains. Numerical examples are presented to illustrate the results.
  %% Another version
 In this paper, we discuss various topological and metrical aspects of the set of stabilizing static feedback gains for multiple-input-multiple-output (MIMO) linear-time-invariant (LTI) systems, in both continuous and discrete-time. Recently,  connectivity properties of this set (for continuous time) have been reported in the literature, along with a discussion on how this connectivity is affected by restricting the feedback gain to linear subspaces. 
% We focus our discussion on the set of stabilizing gains for discrete-time LTI systems. 
 We show that analogous to the continuous-time case, one can construct instances where the set of stabilizing feedback gains for discrete time LTI systems has exponentially many connected components.
\end{abstract}


% \begin{AMS}
% \end{AMS}
\end{tcbverbatimwrite}
\input{tmp_\jobname_abstract.tex}
%% ------------------------------------------------------------------
%% END HEADER
%% ------------------------------------------------------------------

\section{Introduction}
\label{sec:intro}
The precise determination of topological and metrical properties of stabilizing feedback gains is of fundamental   importance in classical and modern control theory~\cite{Byrnes_book_1980, Hermann_TAC_1977no, ohara1992differential,Vidyasagar_TAC_1982}.
These properties have recently received renewed interest due to the emergence of
learning type algorithms for control synthesis; in fact, these properties can identify some of the limitations of such algorithms.
For example, when the set of stabilizing feedback gains contains two or more path-connected components, 
the performance of gradient-based algorithms for control synthesis is highly dependent on the selection of the initial gain, and the algorithm may converge to a ``poor'' local minimum.

Despite the longstanding interest in the set of stabilizing feedback gains, their topological and metrical properties have received limited attention in the literature.
For example, differential geometric structures of stable state feedback gains for MIMO systems have been studied using  information geometry in~\cite{ohara1992differential}.
An elegant geometric approach has also shown that the sets of stabilizing feedback gains for continuous and discrete single-input-single-output (SISO) and dyadic systems are bounded by two hyperplanes and three hyperplanes, respectively in~\cite{Prakash_IJC_1983} and~\cite{Fam_TAC_1978}.
Most of the geometric properties of stabilizing gains can not be easily extended from SISO to MIMO systems, as the proposed geometric approach is applicable when the coefficients of the corresponding characteristic polynomials are linear functions of the entries of the feedback gain.
Moreover, it is shown that the set of \emph{stable SISO systems} of order $n$ can not only be non-convex but also disconnected with $n+1$ connected components in the Euclidean topology~\cite{ober1987topology}.
In fact, for a special class of MIMO systems, one can end up with an exponential number of connected components~\cite{feng2018on}.\footnote{That is, exponential in the number of states.}
We note that the aforementioned results, e.g., \cite{feng2018on,ober1987topology},
are for continuous-time systems.

This paper discusses the topological, metrical, and geometric properties of the set of stabilizing feedback gains for both continuous-time and discrete-time MIMO LTI systems. Some results are generalizations from the SISO case, while most of them require different lines of reasoning. In addition, we discuss properties of the set of structured stabilizing controllers. First, we review the connectedness of the set of stabilizing state-feedback gains for a continuous MIMO system, as recently reported in~\cite{feng2018on} and discuss their topological and metrical properties: open, contractible and unbounded. In the meantime, the set of stabilizing output-feedback gains is shown to be open, unbounded, and in general, not connected.\footnote{Some of our results are presented as ``observations'' as we suspect they might have been observed in the earlier system theory literature.}

It should be noted that the separate treatment for continuous and discrete time systems is warranted;
in fact, in contrast to the folklore expectation of unified properties for continuous and discrete time systems, there are counterexamples to show that the analogies between the two are
far from complete~\cite{Jonckheere_TAC_1981, Wimmer_JMSEC_1995}.
The distinct difference between continuous and discrete LTI systems might be due to the fact that the generalized bilinear transform has poles and thus not continuous~\cite{Hung_LAA_1998,Mehrmann_LAA_1996}. Therefore, generalizing the proposed topological properties of the set of stabilizing feedback gains from continuous LTI systems to discrete ones is not straightforward. This is especially the case for the structured case. For example, it was proposed~\cite{feng2018on} that by employing Schwarz matrices, one may generate an instance where the set of structured Hurwitz stabilizing feedback gains has exponentially many path-connected components. However, as far as we know, there is no counterpart of this phenomena  for the set of Schur stable matrices. Nevertheless, we show that for discrete MIMO systems, the set of stabilizing state-feedback system is open, unbounded and contractible in the Euclidean topology. For output-feedback system, the set is open, but could be either bounded or unbounded and not path-connected. Furthermore, %it's reported in \cite{feng2018on} that the set structured feedback gains could have exponentially many connected components in continuous LTI MIMO systems. 
we present a simple construction to demonstrate that the number of connected components also depends exponentially on the number of states for discrete LTI MIMO systems for structured feedback gains. \par 
%when the structured stabilization problem is under consideration. \par
The paper is structured as follows: \S\ref{sec:notations} provides preliminary background and notation; \S\ref{sec:hurwitz} discusses various topological and metrical properties of the set of Hurwitz stabilizing set; \S\ref{sec:schur} is devoted to the set-theoretic properties of the set of Schur stabilizing feedback gains for discrete LTI systems; \S\ref{sec:conclusion} provides concluding remarks.
\section{Notation and Preliminaries}
\label{sec:notations}
% Linear Algebra
We denote by ${\bb M}_n(\bb R)$ the set of $n \times n$ real matrices and ${\bb {GL}}_n(\bb R)$ as its subset of invertible matrices; $\chi_A$ denotes the characteristic polynomial of a square matrix $A \in \bb M_n(\bb R)$; $\bb R^n$ and $\bb C^n$ denote the $n$-dimensional real and complex Euclidean spaces with $n=1$ identified with real and complex numbers. 
 For a vector $v \in \bb R^n$, we use $v_j$ to denote the $j^{\text{th}}$ entry of $v$, where $v = (v_1, \dots, v_n)^\top$. %\par
The spectrum of a matrix $M$, denoted by $\text{Sp}(M)$ consists  of $n$ complex numbers $\{\lambda_1, \dots, \lambda_n\}$, where each eigenvalue is repeated by its multiplicity. Mind that $\text{Sp}(M)$ is not a well-defined object in $\bb C^n$ as we do not impose a natural ordering amongst the complex $n$-tuple. Thus if $\text{Sp}(M) = \{\lambda_1, \dots, \lambda_n\}$ with each $\lambda_j \in \bb C$, $\sigma \text{Sp}(M) = \{\lambda_{\sigma(1)}, \dots, \lambda_{\sigma(n)}\}$ denotes the same set of eigenvalues of $M$ for every permutation $\sigma$ in the permutation group $S_n$. Hence $\text{Sp}(M)$ is more naturally viewed as an element of the quotient space $\bb C^n/S_n$, where the underlying equivalence relation $u \sim v$ is via,
\begin{align*}
  u = (u_1, \dots, u_n)^\top = (v_{\sigma(1)}, \dots, v_{\sigma(n)})^\top,
\end{align*}
for some $\sigma \in S_n$;
endow this quotient space $\bb C^n/S_n$ with a quotient topology induced by the canonical projection $\pi:~\bb C^n \to \bb C^n/S_n$. \par
The following result will subsequently be used.
\begin{theorem}{(Theorem $5.2$ in~\cite{serre2002matrices})}
  \label{thrm:continuous_spectrum}
  The map $M \mapsto \text{Sp}(M)$ is continuous.
\end{theorem}
A matrix $M \in \bb M_n(\bb R)$ is Hurwitz stable if $\max {\bf Re} [\text{Sp}(M)] < 0$ and $M$ is Schur stable if $\rho(M) < 1$, where $\rho(M)$ denotes the spectral radius of $M$.\footnote{$\max$ and $\bf Re$ denote the usual set functions.}
% Complex analysis
We denote the open unit disk of $\bb C$ by $\bb D = \{z \in \bb C: |z| < 1\}$ and the left half plane by $\bb H_{-} = \{z \in \bb C: \text{\bf Re}(z) < 0\}$; $\bb H^{n}_{-}$ will be the $n$-dimensional version of $\bb H_{-}$. The notation $|z|$ denotes the modulus of the complex number $z \in \bb C$ and $\bar{z}$ denotes the complex conjugate of $z \in \bb C$. We use $\bb C[z]$ and $\bb R[z]$ to denote polynomials with complex coefficients and real coefficients, respectively, where $z$ is the corresponding indeterminant. For a polynomial $p(z) \in \bb C[z]$ or $p(z) \in \bb R[z]$, we use $p'(z)$ to denote its derivative with respect to $z$.\par
We consider a continuous LTI MIMO system, 
\begin{align}
  \label{eq:continuous_system}
  \dot{x}(t) &= A x(t) + B u(t), \\
  y(t) &= C u(t),
\end{align}
and discrete LTI MIMO system, 
\begin{align}
  \dot{x}(k) &= A x(k) + B u(k), \\
  y(k) &= C u(k),
\end{align}
where $A \in \bb M_n(\bb R)$, $B\in \bb M_{n \times p}(\bb R)$ and $C \in \bb M_{p \times n}(\bb R)$. We will abbreviate a system by the tripe $(A, B, C)$. We say that a system is controllable and observable if it satisfies the Kalman Rank Condition, namely, $\text{rank}([B, AB, \dots, A^{n-1}B]) = n$ and $\text{rank}([C^\top, C^\top A, \dots, C^\top A^{n-1}]^\top) = n$, respectively~\cite{zabczyk2009mathematical}. For the problem of designing static feedback gains, we are interested in the feedback gain $K \in \bb M_{m \times n}(\bb R)$ with $u(t) = K x(t)$; in terms of designing static output feedback controller, we are interested in the feedback gain $K \in\bb  M_{m \times p}(\bb R)$ with $u(t) = K y(t) = K C x(t)$. For a controllable and observable triplet $(A, B, C)$, we denote the set of \emph{Hurwitz stabilizing output-feedback gains} by
\begin{align}
  \label{eq:continuous_H}
  \ca H &= \{K \in \bb M_{m \times p}(\bb R): A-BKC \text{ is Hurwitz stable}\},
          \end{align}
and the set of \emph{Schur stabilizing output-feedback gains} by
\begin{align}
  \label{eq:discrete_S}
  \ca S &=\{K \in \bb M_{m \times p}(\bb R): A-BKC \text{ is Schur stable}\}.
\end{align}

When we are concerned with a state-feedback system with same parameters $(A, B)$, $\ca H_s$ and $\ca S_s$ are defined as,
\begin{align}
  \label{eq:continuous_H_s}
  \ca H_s &= \{K \in \bb M_{m \times n}(\bb R): A-BK \text{ is Hurwitz stable}\},\\
  \label{eq:discrete_S_s}
  \ca S_s &= \{K \in \bb M_{m \times n}(\bb R): A-BK \text{ is Schur stable}\},
\end{align}
where we have used the subscript $s$ to denote the state-feedback controller.
\begin{remark}
  We note that if $C \in GL_n(\bb R)$, then $\ca H = C^{-1} \ca H_s$, namely the set is precisely $\ca H_s$ under a change of coordinates. All topological and metrical properties will be identical to $\ca H_s$ in the case $C \in GL_n(\bb R)$. Hence, it is natural to categorize $(A, B, C)$ with $C \in GL_n(\bb R)$ as ``state-feedback'' systems.
\end{remark}
% We shall observe a relation between $\ca H$ and $\ca H_s$ and the similar relation between $\ca S$ and $\ca S_s$. This relation will be used in several occasions of our proofs.
% \begin{observation}
  % \label{obs:state-output-relation}
  % For a controllable and observable system $(A, B, C)$,
  % \begin{align*}
    % \ca H = \{ K \in \bb M_{m \times p}(\bb R): (KC) \cap \ca H_s \neq \emptyset\}.
  % \end{align*}
% \end{observation}

In the analysis of SISO systems~\cite{bu2019siso}, the canonical controller form is proven to be useful in simplifying proofs in many cases. For MIMO systems, we shall employ the Brunovsky controller form~\cite{brunovsky_K_1970}. Any controllable system pair $(A, B)$ can be transformed into the Brunovsky form through a change of basis and translation. So the set of Hurwitz (Schur) stabilizing controllers for a Brunovsky pair amounts to a change of basis of the original set of Hurwitz stabilizing set, as we shall show next.
\par
For a controllable system pair $(A, B)$, if $\text{rank}(B) = r \le m$, there exists $T \in GL_n(\bb R)$, $V \in GL_m(\bb R)$ and $F \in \bb M_{m \times n} (\bb R)$ such that,
\begin{align*}
  (A^{\flat}, B^{\flat}) = (T(A+BF)T^{-1}, TB),
\end{align*}
where $A^{\flat}$ and $B^{\flat}$, respectively, admit the form,
  \begin{align*}
    A^{\flat} = \begin{pmatrix}
      A_{k_1} & 0 & \dots & 0 \\
      0 & A_{k_2} & \dots & 0 \\
      \vdots & \vdots & \ddots & \vdots \\
      0 & 0 & \dots & A_{k_r} ,
    \end{pmatrix},
  \end{align*}
  where each $A_{k_j} \in \bb M_{k_j \times k_j} (\bb R)$ is of the form,
  \begin{align*}
    \begin{pmatrix}
      0 & 1 & 0 &\cdots & 0 \\
      0 & 0 & 1 & \cdots & 0 \\
      \vdots & \vdots & \ddots & \ddots & \vdots \\
      \cdots & \cdots & \cdots & 0 & 1 \\
      0 & 0 & 0 & \cdots & 0
    \end{pmatrix},
    \end{align*}
    such that $k_1 \ge k_2 \ge \dots \ge k_r$ with $k_1 + k_2 + \dots + k_r = n$ and $B^{\flat}$ assumes the structure,
    \begin{figure}[H]
      \center
    \begin{tikzpicture}[mymatrixenv]
    \matrix[mymatrix] (m)  {
       0 & 0 & \cdots & 0 & 0 & \cdots & 0 \\
 \vdots & \vdots & \ddots & \vdots & \vdots & \ddots & 0 \\
 1 & 0 & \cdots & 0 & 0 & \cdots & 0\\
       0 & 0 & \cdots & 0 & 0 & \cdots & 0 \\
 \vdots & \vdots & \ddots & \vdots & \vdots & \ddots & 0 \\
 0 & 1 & \cdots & 0 & 0 & \cdots & 0\\
 \vdots & \vdots & \vdots & \vdots & \vdots & \vdots & \vdots \\
       0 & 0 & \cdots & 0 & 0 & \cdots & 0 \\
 \vdots & \vdots & \ddots & \vdots & \vdots & \ddots & 0 \\
 0 & 0 & \cdots & 1 & 0 & \cdots & 0\\
    };
    \mymatrixbracetop{1}{4}{$r$ columns}
    \mymatrixbracetop{5}{7}{$m-r$ columns}
    \mymatrixbraceright{1}{3}{$k_1$ rows}
    \mymatrixbraceright{4}{6}{$k_2$ rows}
    \mymatrixbraceright{8}{10}{$k_r$ rows} 
\end{tikzpicture}
\end{figure}
If the set of Hurwitz stabilizing state-feedback gains for $(A^{\flat}, B^{\flat})$ is denoted by $\ca H_s^{\flat}$ (respectively, the set of Schur stabilizing state-feedback gains is denoted by $\ca S_s^{\flat}$), we observe that $\ca H_s^{\flat}$ amounts to a change of coordinates and translation of $\ca H_s$.
 \begin{observation}
   \label{obs:feedback_equiv}
  Suppose that $(A, B)$ is controllable and $(A^{\flat}, B^{\flat})$ is the corresponding Brunovsky form.
  % , i.e.,
%\begin{align*}
%  (A^{\flat}, B^{\flat}) = (T(A+BF)T^{-1}, TB).
%\end{align*}
 Then,
 \begin{align*}
   \ca S_s &= \{V \hat{K} T^{-1}-F: \hat{K} \in \ca S_s^{\flat}\}, \\
   \ca H_s &= \{V \hat{K} T^{-1}-F: \hat{K} \in \ca H_s^{\flat}\}.
 \end{align*}
 \end{observation}
 \begin{proof}
   We prove the relation between $\ca S^{\flat}$ and $\ca S$; the proof for ${\ca H^{\flat}}$ and $\ca H$ can proceed similarly. 
   
 By definition of the feedback equivalence, there is $T \in GL_n(\bb R)$, $V \in GL_m(\bb R)$ and $F \in \bb M_{m \times n} (\bb R)$ such that,
 \begin{align*}
   {A^{\flat}} = T (A+BF) T^{-1}, \quad {B^{\flat}} = TBV.
 \end{align*}
 If $\hat{K} \in {S}$, i.e., $\rho({A}^{\flat}-{B}^{\flat}\hat{K}) < 1$, it follows that,
 \begin{align*}
   \rho({A^{\flat}}-{B^{\flat}}\hat{K}) &= \rho(T(A+BF - BV{K}T)T^{-1}) \\
   &= \rho(A-B(V{K}T^{-1} - F)) < 1.
 \end{align*}
This implies that $\ca S = V \ca S^{\flat} T^{-1} - F \coloneqq \{V {K} T^{-1}-F: {K} \in \ca S^{\flat}\}$. 
 \end{proof}
 \begin{remark}
   Since change of coordinates and translation are diffeomorphic in $\bb M_{m \times n}(\bb R)$, the above observation suggests that the topological properties of stabilizing feedback gains are identical between $(A, B)$ and $(A^{\flat}, B^{\flat})$. This observation has many consequences; for example, when $(A, B)$ is in the Brunovsky form, $A-BK$ has special structure and this structure will be useful in proving that $\ca S_s$ is contractible and regular open (see Lemmas~\ref{lemma:regular_open_S},~\ref{lemma:contractible_S}, and~\ref{lemma:connected_S}).
 \end{remark}
%We shall now observe in assuming $B$ of full column rank, we lose no generality over the topological properties of the sets of Hurwitz and Schur stabilizing feedback gains.
We now observe that the topological properties of the sets of Hurwitz and Schur stabilizing feedback gains are independent of the column rank of $B$.
\begin{observation}
  \label{obs:full_column_rank}
  Suppose that $(A,B)$ is a controllable pair and $(A^{\flat}, B^{\flat})$ is the corresponding Brunovsky form. If $\text{rank}(B)= r < m$, we define $\hat{B}^{\flat} \in \bb M_{n \times r}( \bb R)$ by,
    \begin{figure}[H]
      \center
    \begin{tikzpicture}[mymatrixenv]
    \matrix[mymatrix] (m)  {
       0 & 0 & \cdots & 0  \\
 \vdots & \vdots & \ddots & \vdots  \\
 1 & 0 & \cdots & 0 \\
       0 & 0 & \cdots & 0  \\
 \vdots & \vdots & \ddots & \vdots  \\
 0 & 1 & \cdots & 0 \\
 \vdots & \vdots & \ddots & \vdots  \\
       0 & 0 & \cdots & 0  \\
 \vdots & \vdots & \ddots & \vdots  \\
 0 & 0 & \cdots & 1 \\
    };
    \mymatrixbracetop{1}{4}{$r$ columns}
    \mymatrixbraceright{1}{3}{$k_1$ rows}
    \mymatrixbraceright{4}{6}{$k_2$ rows}
    \mymatrixbraceright{8}{10}{$k_r$ rows}
\end{tikzpicture}
\end{figure}
and let $\hat{\ca S}^{\flat}$ denotes the set of Schur stabilizing controllers for $({A}^{\flat}, \hat{B}^{\flat})$. Then $\ca S$ is diffeomorphic to $\hat{\ca S}^{\flat} \times \underbrace{\bb R^n  \times \dots \times \bb R^n}_{m-r}$.
\end{observation}
\begin{proof}
  It suffices to observe that $\ca S^{\flat}$ is exactly $\hat{\ca S}^{\flat} \times \underbrace{\bb R^n  \times \dots \times \bb R^n}_{m-r}$.
\end{proof}
 According to Observation~\ref{obs:full_column_rank}, without loss of generality, we may consider a full column rank matrix $B$ in studying the topological properties of the sets of Hurwitz and Schur stabilizing feedback gains.
 \subsection{On structured feedback gains}
 \label{sec:structured}
 In decentralized control systems, structured feedback gains, reflecting the underlying interaction network are of particular interest. For example, if the underlying interaction network is modeled by a communication graph $\ca G = (V, E)$ and only a subset of agents are accessible to be controlled upon and the control law must only utilize the information of an agent and its neighbors, the feedback gains must have a zero pattern that is compatible with this communication graph, i.e., $K_{ij} = 0$ if $(i,j) \notin E(\ca G)$. More generally, if $\ca U \subseteq \bb M_{m \times p}(\bb R)$ is a linear subspace, the sets of structured output-feedback stabilizing feedback gains is given by
 \begin{align*}
   \ca K_{\ca H} &= \{K \in \ca U: A-BKC \text{ is Hurwitz stable}\}, \\
   \ca K_{\ca S} &= \{K \in \ca U: A-BKC \text{ is Schur stable}\}.
 \end{align*}
 Accordingly, the set of structured state-feedback stabilizing feedback gains are given by,
\begin{align*}
   \ca K_{\ca H_s} &= \{K \in \ca U_s: A-BK \text{ is Hurwitz stable}\}, \\
   \ca K_{\ca S_s} &= \{K \in \ca U_s: A-BK \text{ is Schur stable}\},
\end{align*}
where $\ca U_s$ is a linear subspace of $\bb M_{n \times m}(\bb R)$. \par
%We shall \emph{emphasize} in the scenario\footnote{Indeed, the most widely used model in networked systems.} that the interaction network is modeled by a graph $\ca G=(V, E)$, the sensible case should be such that $B$ is a diagonal matrix. If $B$ has full dimension, i.e., $B \in \bb M_{n \times n}(\bb R)$, this means we can control all the agents in the network and it will be lose no generality in assuming $B=I$. If only a subset can be controlled, i.e., $B \in \bb M_{n \times m}(\bb R)$ with $m < n$, by permutating the agents, we may assume
We shall {emphasize} that in this scenario, the interaction network is modeled by a graph $\ca G=(V, E)$. Each agent can only have direct control over its own dynamics, using information from own sensors and from communicating with neighboring agents, i.e., $B$ has a diagonal structure.
If all agents have control over their own dynamics, without loss of generality, we can assume that $B=I$. In the case that only a subset of agents has direct control over their respective dynamics, without loss of generality, by permuting the agents, we can assume that $B$ has the form, 
\begin{align*}
B=
  \begin{pmatrix}
    I_{m \times m} \\
    \bf 0_{(n-m) \times m}
  \end{pmatrix}.
 \end{align*}
 
 Depending on the structure of the linear subspace $\ca U$, the sets of state-feedback stabilizing gains, $\ca K_{\ca H_s}$ and $\ca K_{\ca S_s}$ will no longer be path-connected. It has been reported in~\cite{feng2018on} that $\ca K_{\ca H}$ can have exponentially many connected components  (in the dimension of the underlying state). In this direction, a sufficient condition ($B=I$ and $C \in GL_n(\bb R)$) has been proposed in~\cite{feng2018on} in order to guarantee that $\ca K_{\ca H_s}$ is connected. We shall review this result and provide another construction to show the exponential dependence of number of connected components on the state dimension. Moreover, we provide results pertaining to the properties of $\ca K_{\ca S}$.
\section{Properties of Hurwitz Stable Feedback Controllers $\ca H$}
\label{sec:hurwitz}
In this section, we shall observe some of the properties of the sets $\ca H$ and $\ca H_s$:
\begin{enumerate}
\item $\ca H$ and $\ca H_s$ are both open in the Euclidean topology.
  \item $\ca H_s$ is contractible\footnote{Note that every contractible set is simply connected and path-connected.} while $\ca H$ is not connected in general.
  \item $\ca H_s$ is unbounded while $\ca H$ could be either bounded or unbounded. We also observe sufficient conditions under which $\ca H$ is unbounded.
\end{enumerate}
%We first observe $\ca H$ and $\ca H_s$ are open sets.
\begin{lemma}
  \label{lemma:open_H}
  $\ca H$ is open in $\bb M_{m \times p}( \bb R)$ and $\ca H_s$ is open in $\bb M_{m \times n}(\bb R)$.
\end{lemma}
\begin{proof}
  We note that the map $\tilde{\upsilon}: \bb C^n/\bb S_n \to \mathbb R$ given by $v \mapsto \max_i \text{\bf Re}(v_i)$ is continuous where $v_i$ denotes the $i$th component of $v$. It is clear that the map $\upsilon: \bb C^n \to \bb R$ given by $v \mapsto \max_i \text{\bf Re}(v_i)$ is continuous. Based on the properties of quotient topology (Theorem $3.73$ in~\cite{lee2010introduction}), $\tilde{\upsilon}$ is the unique continuous map such that $\upsilon = \tilde{\upsilon} \circ \pi$:
\begin{figure}[H]
        \centering
        \begin{tikzcd}
          \bb C^n \arrow[d, "\pi"] \arrow[rd, "\upsilon"] & \\
          \bb C^n_{\text{*}}/S_n \arrow[r, "\tilde{\upsilon}"] & \bb R;
        \end{tikzcd}
      \end{figure}
  We observe that the map $f: \bb M_{m \times p} (\bb R) \to \bb R$ defined by,
  \begin{align*}
    &K \mapsto A-BKC \mapsto \text{Sp}(A-BKC) \mapsto \text{\bf Re}(\text{Sp}(A-BKC)) \\
    &\mapsto \max\left(\text{\bf Re}(\text{Sp}(A-BKC))\right),
  \end{align*}
  is continuous as a composition of continuous maps. Since $\mathcal H = f^{-1}((-\infty, 0))$, $\mathcal H$ is open.
 
For $\ca H_s$, we only need to observe that the function $f_s: \bb M_{m \times n}(\bb R) \to \bb R$ given by, \begin{align*}
    &K \mapsto A-BK \mapsto \text{Sp}(A-BK) \mapsto \text{\bf Re}(\text{Sp}(A-BK)) \\
    &\mapsto \max\left(\text{\bf Re}(\text{Sp}(A-BK))\right),
  \end{align*}
is continuous.
\end{proof}
Next, we prove that the set $\ca H_s$ is contractible. We first observe that the linear matrix inequality (LMI) parametrization of static feedback gains~\cite{dullerud2013course} is a diffeomorphism between $\ca H_s$ and a convex set.
In this direction, suppose that $Q \succ 0$ is a positive definite matrix. By Lyapunov matrix theory~\cite{dullerud2013course}, $K \in \ca H_s$ if and only if there exists $P \succ 0$ such that,
\begin{align}
\label{eq:continuous_lyapunov}
(A-BK)P + P(A-BK)^\top + Q =0.
\end{align}
Consider the change of variable $Y= KP$ that yields,
\begin{align}
  \label{eq:myeq1_h}
  AP + PA^\top - BY - Y^\top B^\top + Q = 0.\end{align}
We denote by $\ca L$ as the solution set $(P, Y)$ of~\eqref{eq:myeq1_h}, i.e.,
\begin{align*}
  \ca L = \{(P, Y): P \succ 0, AP + PA^\top - BY - Y^{\top} B^{\top} + Q = 0\}.
\end{align*}
\begin{lemma}
  \label{lemma:diffeo_H}
  The map $\varphi_Q: \ca L \to \ca H_s$ is a diffeomorphism.  
\end{lemma}
\begin{proof}
  By definition, $\varphi_Q$ is surjective. In order to prove that it is bijective, it suffices to show that for every $K \in \ca H_s$, there is a unique pair $(P, Y) \in \ca L$ such that $\varphi_Q((P, Y)) = K$. For every $(P, Y) \in \varphi^{-1}_{Q}(K)$, we must have, 
\begin{align*} AP + PA^\top -BY-Y^\top B^\top + Q =0
\quad \lra \quad AP + PA^\top -BKP-PK^\top B^\top + Q = 0. \end{align*} Note that the solution $P$ of \eqref{eq:continuous_lyapunov} is unique if $A-BK$ is stable. Hence, $\varphi_Q$ is bijective. The map $(P, Y) \to YP^{-1}$ and $K \mapsto (P(K), KP)$ are both $C^{\omega}$ (real analytic). Thereby, $\varphi_Q$ is a diffeomorphism. 
\end{proof}
As diffeomorphism preserves topological properties, we immediately conclude the following.
\begin{lemma}
  \label{lemma:contractible_H}
  The set $\ca H_s$ is contractible.
\end{lemma}
%%% Old proof
% \begin{proof}
  % By Observation~\ref{obs:full_column_rank}, we may assume $(A,B)$ is in Brunovsky form and $B$ is of full column rank. Let $\ca P \subseteq \bb S_n^+$ be given by
% \begin{align*}
  % \ca P=\{P \in \bb S_n^+: (I - BB^{\dagger})(AP + PA^{\top} + I)(I-BB^{\dagger}) = 0\}.
% \end{align*} and $\ca Q$ be the subset of skew-symmetric real matrices given by
% \begin{align*}
  % \ca Q = \{Q \in \bb M_n(\bb R): Q=-Q^{\top}, BB^{\dagger} Q BB^{\dagger} = Q\}.
% \end{align*}
% By Theorem $3.1$ in~\cite{ohara1993geometric}, the map $\varphi: \ca H \to \ca P \times \ca Q$ given by
% \begin{align*}
  % K \mapsto -B^{\dagger}(AP+PA^{\top} + I)(I - \frac{1}{2}BB^{\dagger})P^{-1} - B^{\dagger}QP^{-1}
% \end{align*}
% is diffeomorphic. We observe the set $\ca P$ and $\ca Q$ are both convex since they are solutions to linear equations. As convex set is contractible, $\ca H$ is contractible.
% \end{proof}
\begin{proof}
  It suffices to observe that $\ca L$ is a convex set and thus contractible; hence, $\ca H_s$ is contractible.
\end{proof}
\begin{remark}
  In~\cite{ohara1993geometric}, another diffeomorphism was proposed under the assumption that $B$ has full column rank.\footnote{It is possible to conclude the set is contractible if combining with Observation~\ref{obs:full_column_rank}.} In~\cite{feng2018on}, the same LMI formulation of $\ca H_s$ was employed to show that the set $\ca H_s$ is path-connected, where it was observed that $\ca H_s$ is the continuous image of $\ca L$ under $\varphi_Q$.
\end{remark}

For output-feedback system $(A, B, C)$, the set $\ca H$ will be no longer connected and this is not even true for SISO systems~\cite{bu2019siso}. \newline 

In the SISO case~\cite{bu2019siso}, the set $\ca H_s$ for SISO systems is regular open and the boundary can be characterized. However for MIMO systems, it becomes rather intricate to determine whether the boundary of $\ca H_s$ is exactly $$\ca B_s=\{K \in \bb M_{m \times n}: \max \text{\bf Re} (\text{Sp}(A-BK)) = 0 \}.$$ We observe a sufficient condition under which $M \in \ca B_s$ is on the boundary $\partial \ca H_s$. Define $g: \bb M_{m \times n}(\bb R) \to \bb R^n$ by $K \mapsto \chi_{A-BK} \cong \bb R^{n}$, which maps $K \in \bb M_{m \times n}(\bb R)$ to the coefficients of the characteristic polynomial of $A-BK$.\footnote{$\chi_{A-BK} \cong \bb R^n$ means that we are identifying the characteristic polynomial with $\bb R^n$ by the natural bijection between a monic $n$th degree polynomial and its coefficients.}
\begin{proposition} \label{prop:suff_cond_1}
  For $M \in \ca B_s$, if $\rank(Dg(M)) = n$ then $M \in \ca H_s$.
\end{proposition}
\begin{proof}
  If $\text{rank}(Dg(M)) = n$, then by Constant Rank Theorem~\cite{tu2010introduction}, there are open neighborhoods of $U$ about $M$, and $V$ about $g(M)$, and diffeomorphisms $\varphi$ of $U$ sending $M$ to the origin of $\bb M_{m \times n}(\bb R)$, $\phi$ of $V$ sending $g(M)$ to the origin of $\bb R^n$, such that $\phi \circ g \circ \varphi^{-1}$ is a projection, i.e.,
\begin{align*}
  \phi \circ g \circ \varphi^{-1}(M_{11}, \dots, M_{1n}, M_{21}, \dots, M_{2n}, \dots, M_{mn} ) = (M_{11}, \dots, M_{1n}),
\end{align*}
where $M_{ij}$'s denote the entries of $M$. Now it is clear that we may perturb the entries of $M$ to get a sequence in $\ca H_s$ converging to $M$ and a sequence in $\ca H_s^{c}$ converging to $M$.
\end{proof}
It should be noted that Proposition~\ref{prop:suff_cond_1} only provides a sufficient condition since the differential does not always have full rank. For example, if we take $(A, B)$ in the Brunovsky form with $A \in \bb M_4(\bb R)$ having $2$ blocks of size $2 \time 2$, then $A-BK$ has the form,
\begin{align*}
  \begin{pmatrix}
    0 & 1 & 0 & 0 \\
    a_1 & a_2 & a_3 & a_4 \\
    0 & 0 & 0 & 1 \\
    b_1 & b_2 & b_3 & b_4
  \end{pmatrix}.
\end{align*}
Direct computation reveals that rank of $Dg$ will be greater than $2$; as such, $M \in \ca B_s$ and $\text{rank}(Dg(M)) < 4$. It is thereby unclear whether $M$ is on the boundary.

We now proceed to examine the boundedness of $\ca H_s$ and $\ca H$. 
%We first observe that $\ca H_s$ is unbounded.
\begin{observation}
  $\ca H_s$ is unbounded.
\end{observation}
\begin{proof}
  This is a consequence of Pole Shifting Theorem~\cite{zabczyk2009mathematical}: for every $n$-tuple numbers $\{-j, \dots, -j\}$ with $j \in \bb R$, there exists $K_j$ such that the spectrum of $A-BK_j$ is exactly $\{-j, \dots, -j\}$.
\end{proof}
%In SISO case, the boundedness of $\ca H$ is determined by the sign of components of $c$.
In the output feedback case, the set can be either bounded or unbounded. We first provide 
an example where $\ca H$ is bounded.
\begin{example1}
  Consider a controllable and observable triplet given by,
\begin{align*}
A =
  \begin{pmatrix}
   0 & 1 & 0 & 0 & 0 & 0 & 0 & 0\\ 
   0 & 0 & 1 & 0 & 0 & 0 & 0 & 0\\ 
   0 & 0 & 0 & 1 & 0 & 0 & 0 & 0\\ 
   0 & 0 & 0 & 0 & 1 & 0 & 0 & 0\\ 
   0 & 0 & 0 & 0 & 0 & 1 & 0 & 0\\ 
   0 & 0 & 0 & 0 & 0 & 0 & 1 & 0\\ 
   0 & 0 & 0 & 0 & 0 & 0 & 0 & 1\\ 
   -1 & -8 & -28 & -56 & -70 & -56 & -28 & -8\\ 
  \end{pmatrix},\,
B = \begin{pmatrix}
  0 & 0\\
  0 & 0\\
  0 & 0\\
  1 & 0\\
  0 & 0\\
  0 & 0\\
  0 & 0\\
  0 & 1\\
\end{pmatrix}, \, C=\begin{pmatrix}
  1 \\ -1 \\ 0 \\ 0 \\ 0 \\ 0 \\ 1 \\ -1
\end{pmatrix}^\top.
\end{align*}
The gain $K$ will then be parameterized by two scalars. The characteristic polynomial of the closed-loop system is given by,
\begin{align*}
  \chi_{A-BKC} &= z^8 + (8 - k_2)z^7 + (56k_1+28+ k_2)z^6 + (-29k_1+56)z^5 + (-27k_1 +70)z^4 \\
&\quad + (-27k_1+56)z^3 + (-29k_1 + 28)z^2 + (-k_2 -14k_1 +8)z + 1 + 70k_1+k_2.
\end{align*}
If $K \in \ca H$, the coefficients of $\chi_{A-BKC}$ are necessarily positive~\cite{zabczyk2009mathematical}; from which, we conclude that both $k_1$ and $k_2$ are bounded.
\end{example1}
We shall now observe sufficient conditions under which $\ca H$ is unbounded.
First, if either $B$ or $C$ do not have full column/row rank, then the set $\ca H$ is unbounded.
\begin{observation}
  If $\text{rank}(B) < m$ or $\text{rank}(C) < p$, then $\ca H$ is unbounded.
\end{observation}
\begin{proof}
If $\text{rank}(B)<m$, the product $B E_c$ has a zero column by performing elementary column operations encoded by $E_c$. Then the corresponding column of $K$ can be arbitrarily chosen without affecting the characteristic polynomial of $A-BKC$.
\end{proof}

% In SISO case, the boundedness of $\ca H$ is determined by the sign of components of $c$.
% However in MIMO case, the coefficients of $\chi_{A-BKC}$ will contain rather complicate linear/nonlinear combinations of the entries of $K$. So there is no transparent way to determine the boundedness. We shall observe a sufficient condition that guarantees the set $\ca H$ is not bounded.
% But we can conclude that for almost every triplet $(A, B, C)$, if $\ca H \neq \emptyset$, then $\ca H$ is not bounded.
% \begin{lemma}
  % Suppose $(A, B, C)$ is controllable and observable with $B$ and $C$ having full rank $m$ and $p$ respectively. For almost every\footnote{``Almost every'' means the property is true except an algebraic (Zariski closed) set.} pairs $(B, C)$, if $\ca H$ is not empty, then $\ca H$ is unbounded.
% \end{lemma}
% \begin{proof}
  % If $\ca H$ is nonempty, then there is a $K \in \ca H$ such that $A-BKC$ is stable and $A-BKC$ has distinct eigenvalues. By first stabilizing the system, we may without loss of generality assume $A$ is stable and has distinct eigenvalues in $(A, B, C)$. Then the Lemma follows from Lemma $2$ in~\cite{davison1975pole} as we can always assign $\min(n, m+p-1)-\min(m, p)+1$ together with eigenvalues 
% \end{proof}
We now observe that if $B$ and $C$ have full rank, and the dimensions $m$ and $p$ are large enough, then $\ca H$ is unbounded.
\begin{observation}
  \label{obs:cont_unbounded_1}
   If $n \le m+p -1$, then $\ca H$ is unbounded.
\end{observation}
\begin{proof}
  By Theorem $3$ in~\cite{kimura1975pole}, for any $n$ tuple of complex numbers $\Lambda$, invariant under complex conjugation, there is $K$ such that the spectrum of $A-BKC$ is arbitrarily close to $\Lambda$. This implies that $\ca H$ is nonempty and unbounded.
\end{proof}
% \begin{proof}
  % For $K_0 \in \ca H$, denote the Jordan canonical form of $A-BK_0C$ by $A-BK_0C = SJS^{-1}$. For $S^{-1}B$, we may multiply an \emph{elementary matrix} $E_c$ on the right such that the first column of $S^{-1}B$ has the form $\mathbf c_{1}(B)= (\underbrace{0, \dots, 0}_{m-1}, *, \dots, *)^{\top}$. Similarly, there is an \emph{elementary matrix} $E_r$ acting on the left such that the first row of $CS$ has the form $\mathbf r_{1}(C) = (\underbrace{0, \dots, 0}_{r-1}, *, \dots, *)$. We observe if $K = \{\bf e_1, \dots, \bf 0, \bf 0\}$,
% \begin{align*}
  % S^{-1}B E_c\begin{pmatrix}
    % 0 & \dots & 0 \\
    % \vdots & \ddots & \dots \\
    % 0 & \dots & 1
  % \end{pmatrix}
                % E_r CS = \begin{pmatrix}
                  % 0 & 0 & \dots & * \\
                  % 0 & 0 & \dots & * \\
                  % 0 & \vdots & \vdots & * \\
                  % 0 & \dots & \dots & 0
                % \end{pmatrix}
% \end{align*}
% has the nonzero entries above the main diagonal, thus
% is upper triangular with diagonal entries to be $0$'s. So if $K_0 \in \ca H_{s}$, 
% $K_0 + j E_cKE_r$ is stabilizing for any $j \in \bb R$ since
% \begin{align*}
  % A-B(K_0 + j E_c K E_r)C &= SJ S^{-1} + S\left(S^{-1}BE_c K E_r CS\right)S^{-1} \\
  % &= S (J + L)S^{-1}.
% \end{align*}Hence $\ca H$ is not bounded.
% \end{proof}
The necessary condition for unboundedness of $\ca H$ in Observation~\ref{obs:cont_unbounded_1} can be relaxed if we only require that the statement holds for ``almost every'' triplet $(A, B, C)$.\footnote{``Almost every'' means that the property is valid except on an algebraic (Zariski closed) set.}
%observe more sufficient conditions which guarantee $\ca H$ is not bounded. %{\color{red}List references.}
\begin{observation}
  Suppose that $(A, B, C)$ is controllable and observable, where $B$ and $C$ have full ranks $m$ and $p$, respectively.
  %\begin{enumerate}
     If $n < mp$, then for almost very controllable and observable triplet $(A, B, C)$, the set $\ca H_s$ is unbounded.
%\item If $A$ is stable with distinct eigenvalues and $m+p-1 \ge n$ and $\ca H \neq \emptyset$, then for almost every pair $(B, C)$, $\ca H$ is not bounded.
 % \end{enumerate}
\end{observation}
\begin{proof}
  If $n < mp$, by Proposition $2.8$ in~\cite{wang1996grassmannian}, for ``almost every'' controllable and observable triplet, we may arbitrarily assign poles by the output feedback gain $K$; the conclusion thus follows.
  %The observations follow from theorems in the literature of pole assignment by output-feedback. If $(A, B, C)$ is controllable and observable with $n < mp$, then for any $n$-tuple of complex numbers invariant under conjugation, we may find a $K$ such that $A-BKC$ has spectrum arbitrarily close to this set (see Theorem ). This clearly implies the unboundedness. The second theorem follows from Lemma $2$ in~\cite{davison1975pole}. If the prescribed conditions are satisfied, we may arbitrarily assign $\min(n, m+p-1)-\min(m,p) -1$ complex eigenvalues together with arbitrarily chosen $\min(m, p)+1$ eigenvalues from $A$.
\end{proof}
%It is reported in~\cite{feng2018on} that $\ca H_s$ is path-connected which has clear algorithmic consequence. For completeness, we include a short proof.
%\begin{lemma}{(Lemma $1$ in~\cite{feng2018on})}
    %$\ca H$ is connected if $C \in GL_n(\bb R)$.
%\end{lemma}
%\begin{proof}
  %It suffices to assume $C=I$.
  %Note $K \in \ca H_s$ if and only if there exists some $P \succ 0$ such that $(A-BK)P + P(A-BK)^T \prec 0$. A change of variable $Y= KP$ yields
%\begin{align}
  %\label{eq:myeq1_h}
  %AP + PA^T + BY + Y^TB^T \prec 0.
%\end{align}
%Observe the solution set $(P, Y)$ of \eqref{eq:myeq1_h} is a convex set and thus path-connected. So $\ca H_s$ is connected as a continuous image under the map $(P, Y) \mapsto YP^{-1}$.
%\end{proof}
\subsection{Connectedness of Structured Hurwitz Stabilizing Feedback Gains}
In designing gradient-based algorithms that evolve directly on the set of structured stabilizing feedback gains, connectedness of this set plays an important role. If the set has several path-connected components, the outcome of gradient-based algorithms will be dependent on the initialization process. A sufficient condition to guarantee connectedness is proposed in~\cite{feng2018on}. For completeness, we provide a transparent proof here which shares the essence of the proof in~\cite{feng2018on}.
\begin{lemma}{(Lemma $2$ in~\cite{feng2018on})}
  \label{lemma:subspace_c}
  Suppose $\ca U$ is a linear subspace in $\bb M_{n \times n}(\bb R)$. If $B = I$ and $I \in \ca U$, then the set $\ca K_{\ca H_s} = \{K \in M_{n \times n}(\bb R): K \in \ca U, K \in \ca H_s\}$ is connected.
\end{lemma}
\begin{proof}
  For $K_1, K_2 \in \ca K$, note that the map $\gamma: [0,1] \to \bb R$ given by $$t \mapsto (1-t)K_1 + tK_2 \mapsto \max \text{\bf Re}(\text{Sp}((1-t)(A-K_1) + t(A-K_2)))$$ is continuous. As $[0,1]$ is compact, $\gamma$ achieves its maximum value at $c$. If $c < 0$, $(1-t)K_1 + tK_2$ is a continuous path between $K_1$ and $K_2$. If $c \ge 0$, we first connect $K_1$ to $K_1+c'I$ and $K_2$ to $K_2 +c'K_2$ with $c' > c$ by convex paths, i.e., $t \mapsto (1-t)K_j + t(K_j + c'I)$. Note that these paths stay in $\ca K_{\ca H_s}$ if $c'> 0$. Then the convex path $t \mapsto (1-t)(K_1 + c'I) + t(K_2 + c'I)$ stays in $\ca K_{\ca H}$ as the maximum of $\gamma': [0, 1] \to \bb R$ defined by
$$t \mapsto \max \text{\bf Re}(\text{Sp}((1-t)(A-K_1 - c'I) + t(A-K_2-c'I)), $$
is exactly $c - c' < 0$. Hence, $K_1 \to K_1 + c'I \to K_2 + c'I \to K_2$ is a continuous path in $\ca K$, where each arrow is connected by a convex path. 
\end{proof}
As discussed in~\cite{feng2018on}, one may generate $2^{n-1}$ connected components in the set of structured stabilizing feedback gains by employing the properties of Schwarz matrix (see Theorem $2$ in~\cite{feng2018on} for details). Here, we present a conceptually simple construction to show the exponential dependence of the number of connected components on the dimension of the state; in fact, this example leads to a lower bound of $2^{\floor{n/2}}$. One salient feature of our construction is that a similar idea can be extended to discrete-time systems (see Lemma~\ref{lemma:exponential_components_discrete}).
% However, the construction presented in~\cite{feng2018on} could not be generalized to discrete-time LTI systems.
\begin{proposition}
  \label{prop:cont_components_2}
  Suppose $\ca U \subseteq \bb M_{2 \times 2}(\bb R)$ is a linear subspace given by
  \begin{align*}
     \ca U = \{U \in \bb M_{2 \times 2}(\bb R): u_{12}=-u_{21}, u_{11}=u_{22}=0\}.
  \end{align*}
  If
  \begin{align*}
    A = \begin{pmatrix}
      -1 & -1 \\
      1 & 0
    \end{pmatrix}, \quad B = I,
  \end{align*}
  then the set $\ca K = \{K \in \bb M_{2 \times 2}(\bb R): K \in \ca U, K \in \ca H_s\}$ has exactly two connected components.
\end{proposition}
\begin{proof}
  We note that for $K \in \ca U$, $A-BK$ has the form,
  \begin{align*}
    \begin{pmatrix}
      -1 & -(1-t) \\
      1-t & 0
    \end{pmatrix},
  \end{align*}
  with the characteristic polynomial $\chi_{A-BK}(\lambda)= (\lambda+1)^2 \lambda + (1-t)^2$. By Routh-Hurwitz criterion, $\chi_{A-BK}$ is stable if and only if $t \neq 1$. Hence, $\ca K$ has two connected components: $(-\infty, 0)$ and $(0, \infty)$.
\end{proof}
We now construct an instance of the structured synthesis problem that leads to $2^{\floor{n/2}}$ connected components in the set of stabilizing feedback gains. In the following, we will provide an explicit construction for the case where $n$ is even. When $n$ is odd, we consider a block diagonal matrix for $A$ with an even dimension of $n-1$ and a constant $-1$ on its diagonal.
\begin{lemma}
    For $n=2k$, suppose that
    \begin{align*}
      A = \begin{pmatrix}
        -1 & -1 & & & & & & & \\
        1 & 0 & & & & & & &\\
        & & -1 & -1 & & & && \\
        & & 1 & 0 & & & && \\
        & & & & \ddots & & && \\
        & & & & & \ddots &  & & \\
        & & & & & &  & -1&-1 \\
        & & & & &  &  & 1&0 \\
      \end{pmatrix}, \qquad B = I,
    \end{align*} and let $\ca U \subseteq \bb M_{n \times n}(\bb R)$ be a linear subspace defined by
    \begin{align*}
      \ca U = \{U \in \bb M_{n \times n}(\bb R): u_{12} = - u_{21}, \dots, u_{2k-1, 2k} = - u_{2k, 2k-1}, \text{ other entries are $0$'s} \}.
    \end{align*}
    Then the set $\ca K =\{K \in \bb M_{n \times n} (\bb R): K \in \ca U, K \in \ca S_s\}$ has exactly $2^{k}$ connected components.
\end{lemma}
\begin{proof}
  We only need to observe that $\ca K$ is the Cartesian product of the intervals obtained in Proposition~\ref{prop:cont_components_2}.
\end{proof}
\section{Properties of Schur Stable Feedback Controllers $\ca H$}
\label{sec:schur}
In this section, we discuss topological properties of Schur stabilizing feedback gains. These properties include:
\begin{enumerate}
  \item The set $\ca S$ and $\ca S_s$ are both open.
    \item The set $\ca S_s$ is contractible and regular open, i.e., $\bar{\ca S}_s^{\circ} = \ca S_s$. In general, the set $\ca S$ is not connected.
      \item $\ca S_s$ is not bounded and $\ca S$ could be either bounded and unbounded.
      \item If $K$ is constrained to linear subspaces, sufficient conditions are proposed to guarantee that the set of structured stabilizing gains is path-connected. We also show that there exists an instance such that the set $\ca S$ has $2^{\floor{\frac{n}{2}}}$ connected components.
\end{enumerate}

For single-input state-feedback systems, given a controllable pair $(A, b)$, it was observed in~\cite{bu2019siso} that the bilinear transformation $z \mapsto (z+1)(z-1)^{-1}$ provides a homeomorphism between the set of Hurwitz stabilizing gains and the set of Schur stabilizing gains. It was also observed in~\cite{bu2019siso} that the bilinear transformation does not provide a homeomorphism between the set of Hurwitz stabilizing output-feedback gains and the set of Schur stabilizing output-feedback gains. For MIMO case, bilinear transformation does not yield a homeomorphism even for state-feedback case. For example, if
\begin{align*}
  A = \begin{pmatrix}
    0 & 1 & 0 & 0 \\
    0 & 0 & 0 & 0 \\
    0 & 0 & 0 & 1 \\
    0 & 0 & 0 & 0
  \end{pmatrix}, \quad B=\begin{pmatrix}
    0 & 0 \\
    1 & 0 \\
    0 & 0 \\
    0 & 1
  \end{pmatrix},
\end{align*}
then $0 \in \ca S_s$. However, under the bilinear transformation $X \mapsto (X-I)^{-1}(X+I)$, $A-B 0$ will be mapped to
\begin{align} \label{eq:closed_A}
  \begin{pmatrix}
    -1 & -2 & 0 & 0 \\
    0 & -1 & 0 & 0\\
    0 & 0 & -1 & -2\\
    0 & 0 & 0 & -1
  \end{pmatrix}.
\end{align}
It is clear there is no $K \in \ca H_s$ such that $A-BK$ will yield the matrix~(\ref{eq:closed_A}). Therefore, it is necessary to study the set-theoretic properties for discrete LTI systems independently. 
%We shall now observe the openness.
\begin{lemma}
  \label{lemma:open_S}
  $\ca S$ is open in $\bb M_{m \times p}(\bb R)$ and $\ca S_s$ is open in $\bb M_{m \times n}(\bb R)$.
\end{lemma}
\begin{proof}
  The proof proceeds similar to the continuous case in Lemma~\ref{lemma:open_H}. We only need to observe that the map $f_{\ca S}: \bb C^n/ S_n \to \bb [0, \infty)$ is continuous by passing to the quotients (see details in Lemma~\ref{lemma:open_H}) and $[0, 1)$ is open in $[0, \infty)$.
\end{proof}
Contrary to the discrete SISO system~\cite{bu2019siso}, the set of stabilizing state-feedback gains for MIMO systems is unbounded.
\begin{observation}
  \label{lemma:bounded}
   $\ca S_s$ is generally unbounded.
\end{observation}
\begin{proof}
  It suffices to assume that $(A, B)$ is in the Brunovsky form since changing of coordinates and translation do not change boundedness. Without loss of generality, if $B$ does not have full column rank, then the statement of this observation is valid since we can choose the last row of $K$ arbitrarily. Now suppose that $B$ has full column rank.
This on the other hand, allows constructing a sequence of stabilizing feedback gains with an unbounded norm. Since $(A, B)$ is controllable, according to Pole-Shifting Theorem, for any polynomial of degree $n^{th}$, there is a $K$ such that $p(t) = \chi_{A-BK}(t)$. By taking $p(t)=\Pi_{i=1}^n (t-\lambda_i)$, where $\lambda_i \in \bb R$ are distinct, there is $K_1 \in \bb M_{m \times n}(\bb R)$ such that
  \begin{align*}
    A - BK_1 = S^{-1} 
    \begin{pmatrix}
      \lambda_1 & 0 & \dots & 0 \\
      0 & \lambda_2 & \dots & 0 \\
      \vdots & \vdots & \ddots & 0 \\
      0 & 0 & \vdots & \lambda_n
    \end{pmatrix} S.
  \end{align*}
    Putting $\tilde{B} = SB$, since $B$ has full column rank, we may choose an (column) elementary matrix $E$ such that $(\tilde{B}E)_{nm} = 0$. Let $K_2 = E(0, 0, \dots, e_m)$, where $0 \in \bb R^m$ and $e_n = (0, \dots, 1)^T \in \bb R^m$. Then $SBK_2 = (0,\dots, 0, \tilde{B}Ee_m )\eqqcolon L$ is upper triangular with diagonals all $0$'s. Letting $K_2' =  K_2 S$, then
  \begin{align*}
    A - B(K_1 + K_2') &= S^{-1} \Lambda S - S^{-1} SB K_2  S = S^{-1} (\Lambda + L) S.
  \end{align*}
  This shows that adding $K_2'$ will not change the eigenvalues of $A-BK_1$. Now we define a sequence of feedback controllers $\{K^n\}_{n=1}^{\infty} \coloneqq \{K_1 + n K_2'\}_{n=1}^{\infty}$. Then clearly $K^n$ is stabilizing and $\|K^n\| \to \infty$ as $n \to \infty$.
\end{proof}
%In the output-feedback case, the condition in Observation~\ref{obs:cont_unbounded_1} holds for $\ca S_s$ but condition in Observation~\ref{obs:cont_unbounded_2} does not extend to discrete case. This is due to the reason that in discrete case the spectrum must be in the unit disk.
%\begin{observation}
  %\label{obs:dis_unbounded}
  %If $n \le m+p-2$, $\ca S_s$ is unbounded.
%\end{observation}
For output feedback gains, $\ca S$ can be either bounded or unbounded.
\begin{example1}
  \label{ex:boundedness_S}
  We provide two examples such that the set $\ca S$ is bounded in the first case and unbounded in the second.
  \begin{enumerate}
  \item Suppose that $(A, B, C)$ is a controllable and observable system specified by,
  \begin{align*}
    A = \begin{pmatrix}
      0 & 1 & 0 & 0\\
      -\frac{1}{2} & -1 & 0 & 1 \\
      0 & 0 & 0 & 1 \\
      0 & 0 & -\frac{1}{2} & -1
    \end{pmatrix}, \; B = \begin{pmatrix}
      0 & 0 \\
      1 & 0 \\
      0 & 0 \\
      0 & 1
    \end{pmatrix}, \; C=\begin{pmatrix}
      1 & 0 & 0 & 0 \\
      0 & 1 & 0 & 0
    \end{pmatrix}.
  \end{align*}
  Then the set $\ca S$ is bounded.
  \item Suppose that $(A, B, C)$ is a controllable and observable system specified by,
    \begin{align*}
    A = \begin{pmatrix}
      0 & 1 & 0 & 0\\
      -\frac{1}{2} & -1 & 0 & 1 \\
      0 & 0 & 0 & 1 \\
      0 & 0 & -\frac{1}{2} & -1
    \end{pmatrix}, \; B = \begin{pmatrix}
      0 & 1 & 0 \\
      1 & 0 & 0\\
      0 & 0 & 0\\
      0 & 0 & 1
    \end{pmatrix}, \; C=\begin{pmatrix}
      0 & 0 & 0 & 1 \\
      0 & 1 & 0 & 0 \\
      1 & 0 & 0 & 0
    \end{pmatrix}.
                  \end{align*}
                  Then $\ca S$ is not bounded.
  \end{enumerate} For part $(a)$, 
  $K \in \bb M_{2 \times 2}(\bb R)$, i.e., $K$ is parametrized by four parameters with $K=\begin{pmatrix}
    k_1 & k_2 \\
    k_3 & k_4
  \end{pmatrix}$. The characteristic polynomial of the closed-loop system is given by,
  \begin{align*}
    \chi_{A-BKC}(t) = t^4 + (-k_2 + 2)t^3 + (-k_1 - k_2 - k_4 + 2)t^2 + (1 - \frac{k_2}{2} - k_1 - k_3)t + \frac{1}{4}- \frac{k_1}{2}.
  \end{align*}
  By Vieta's formula, the coefficients are symmetric polynomials in zeros of the closed-loop system $A-BKC$. It follows that all coefficients of $\chi_{A-BKC}$ are bounded since $A-BKC$ is Schur stable. To see that $\ca S$ is bounded, we only need to observe that $k_1, k_2, k_3, k_4$ are bounded by the triangle inequality. 
  
  For part $(b)$, $K \in \bb M_{3 \times 3}(\bb R)$, i.e., $K$ is parametrized by four parameters with, $$K=\begin{pmatrix}
    k_1 & k_2 & k_3\\
    k_4 & k_5 & k_6 \\
    k_7 & k_8 & k_9
  \end{pmatrix}.$$
We note that $0 \in \ca S$. Putting,
\begin{align*}
  K_c =
  \begin{pmatrix}
              0 & 0 & 0 \\
              0 & 0 & 0 \\
              0 & 0 & c 
  \end{pmatrix},
\end{align*}
with $c \in \bb R$, we observe that,
\begin{align*}
  BK_cC =
  \begin{pmatrix}
              0 & 0 & 0 & 0 \\
              0 & 0 & 0 & 0\\
              0 & 0 & 0 & 0\\
              c & 0 & 0 &0 
  \end{pmatrix}.
\end{align*}
It thus follows that $K_c \in \ca S$ for every $c \in \bb R$.
Hence, in this case, $\ca S$ is not bounded.
  \end{example1}
Next, we shall show that the set of state-feedback stabilizing controllers is contractible and regular open. This follows from an important observation: if $(A, B)$ is in the Brunovsky norm, then under a nonlinear scaling of the entries of $K$, the eigenvalues of the corresponding $A-BK$ will be scaled accordingly. The precise statement of this property
is as follows.
%by Lemma~\ref{lemma:eigenvalue_brunovsky}.
  \begin{lemma}
    \label{lemma:eigenvalue_brunovsky}
    Suppose that $(A, B)$ is in the Brunovsky form where $B$ has full column rank. For every $K \in \bb M_{m \times n}(\bb R)$, denote the spectrum of $A-BK$ by $\sigma(A-BK)= \{\lambda_1, \dots, \lambda_n\}$ and put $(K)_{\alpha}$ as follows: for each $j$, the $j^{\text{th}}$ row of $(K)_{\alpha}$ is given by
%    \begin{align*}
%  (K)_{j, \cdot} = \begin{pmatrix}
%    \alpha^{k_1} (K)_{j1} & \alpha^{k_1-1} (K)_{j2} & \dots & \alpha(K)_{jk_1} & \dots & \alpha^{k_r} (K)_{j(k_{r-1}+1)} & \dots & \alpha (K)_{j k_r}
%  \end{pmatrix}.
%    \end{align*}
        \begin{align*}
  (K)_{j, \cdot} = \Big( &
    \alpha^{k_1} (K)_{j1} , ~\alpha^{k_1-1} (K)_{j2} ,~ \dots , \alpha(K)_{jk_1} , \dots ,~ \alpha^{k_r} (K)_{j(k_{r-1}+1)} ,~ \dots ,~ \alpha (K)_{j k_r}
  \Big).
    \end{align*}
    Therefore, $(\lambda, v) = (\lambda, (v_1, \dots, v_n)^\top)$ is a left eigenvalue-eigenvector pair of $A-BK$,  if and only if $(\alpha \lambda, \tilde{v})$ is a left eigenvalue-eigenvector pair of $A-B(K)_\alpha$, where $r \neq 0$ and
  \begin{align*}
    \tilde{v} = (\underbrace{\alpha^{k_1 - 2} v_1, \alpha^{k_1-3} v_2, \dots, v_{k_1 - 1}, \frac{v_{k_1}}{\alpha}}_{k_1}, \dots, \underbrace{\alpha^{k_r-2} v_{k_1 + \dots + k_{r-1} +1}, \dots, v_{k_r-1}, \frac{v_{k_r}}{\alpha}}_{k_r});
  \end{align*}
    consequently $\sigma(A-B(K)_{\alpha})=\{\alpha\lambda_1, \dots, \alpha \lambda_n\}$.
  \end{lemma}
  \begin{proof}
We note that the first $k_1$ rows of $A-B(K)_\alpha$ have the following form,
%     \begin{figure}[H]
%       \center
%     \begin{tikzpicture}[mymatrixenv]
%     \matrix[mymatrix] (m)  {
%       0 & 1 & \cdots & 0 & \cdots & \cdots & \cdots & \cdots  & \cdots\\
%       0 & 0 &  \cdots & 0  & \cdots & \cdots & \cdots & \cdots  & \cdots\\
%       \vdots & \vdots & \vdots & \ddots & \vdots & \cdots & \cdots   & \cdots & \cdots\\
%       -\alpha^{k_1}(K)_{11} & -\alpha^{k_1-1}(K)_{12} &  \cdots & -\alpha(K)_{1 k_1} & \cdots & -\alpha^{k_r}(K)_{1(n-k_r+1)} & -\alpha^{k_r-1}(K)_{1(n-k_r+2)} &  \cdots & -\alpha(K)_{1n}\\
%     };
% \end{tikzpicture}
% \end{figure}
   \begin{align*}
&\left[\begin{matrix}
       0 & 1 & \cdots & 0 \\
       0 & 0 &  \cdots & 0  \\
       \vdots & \vdots & \vdots & \ddots \\
       -\alpha^{k_1}(K)_{11} & -\alpha^{k_1-1}(K)_{12} &  \cdots & -\alpha(K)_{1 k_1} \\
\end{matrix}\right.\\
&~\qquad
\left.\begin{matrix}
 \cdots & \cdots & \cdots & \cdots  & \cdots \\
  \cdots & \cdots & \cdots & \cdots  & \cdots\\
   \vdots & \cdots & \cdots   & \cdots & \cdots\\
    \cdots & -\alpha^{k_r}(K)_{1(n-k_r+1)} & -\alpha^{k_r-1}(K)_{1(n-k_r+2)} &  \cdots & -\alpha(K)_{1n}\\
\end{matrix}\right] .
  \end{align*}
 Now if $(\lambda, v) = (\lambda, (v_1, \dots, v_n)^\top)$ is a left eigenvalue-eigenvector pair of $A-BK$, it suffices to check the equality of $\tilde{v}^{\top} (A- B(K)_{\alpha})  = \alpha \lambda \tilde{v}^{\top}$ for each component of the vector. We can similarly check the second and other components component of $\tilde{v}$: \begin{align*}
    &\alpha^{k_1 - 2} v_1 - \frac{(K)_{12} v_{k_1}}{\alpha} - \dots - \frac{(K)_{k_r 2} v_{k_r}}{\alpha} \\
    & \quad = \alpha^{k_1 - 1} ( \frac{v_1}{\alpha} - \frac{(K)_{1 2} v_{k_1}}{\alpha} - \dots   - \frac{(K)_{k_r 2} v_{k_r}}{\alpha} ) \\
                                                                                       &\quad = \alpha^{k_1-2} \lambda v_2 
                                                                                        = \lambda \alpha (\alpha^{k_1-3} v_2) \\
    &\quad = \lambda \alpha \tilde{v}_2.
  \end{align*}
  \end{proof}
  Lemma~\ref{lemma:eigenvalue_brunovsky} immediately implies that $\ca S_s$ is contractible.
  \begin{lemma}
    \label{lemma:contractible_S}
    $\ca S_s$ is contractible.
  \end{lemma}
  \begin{proof}
    According to Observation~\ref{obs:full_column_rank}, it suffices to assume that $B$ has full column rank.
    We only need to observe that the map $H: \ca S_s \times [0, 1] \to \ca S_s$ given by $(K, t) \mapsto (K)_{1-t}+t0$ yields a homotopy between the identity map and the constant map $0$ by Lemma~\ref{lemma:eigenvalue_brunovsky}. 
  \end{proof}
  Indeed, Lemma~\ref{lemma:eigenvalue_brunovsky} allows us to characterize the boundary of $\ca S_s$ as well. Putting $\ca B_s = \{K \in \bb M_{m \times n}(\bb R): \rho(A-BK) = 1\}$, we show next that $\ca B_s$ coincides with the boundary $\partial \ca S_s$.\footnote{This is not immediate. We certainly have $\partial \ca S_s \subseteq \ca B_s$; however it is now clear that every point of $\ca B_s$ is a boundary point of $\ca S_s$.}
  \begin{lemma}
    \label{lemma:regular_open_S}
    $\partial \ca S_s = \ca B_s$ and $\ca S_s$ is regular open.
  \end{lemma}
  \begin{proof}
    If $K \in \ca B_s$, it is clear $\{(K)_{1-1/n}\}_{n \ge 2} \subseteq \ca S_s$ is a sequence converging to $K$ and $\{(K)_{1+1/n}\}_{n \ge 2} \subset \ca S_s^c$ converges to $K$ as well. It thus follows that,
    \begin{align*}
      \bar{\ca S_s}^{\circ} = (\ca S_s \cup \ca B_s )^{\circ} = \ca S_s.
    \end{align*}
  \end{proof}
  Next we show that $\ca S_s$ is path-connected. This is straightforward by Lemma~\ref{lemma:contractible_S} since contractible sets are path-connected and simply connected. We shall include an independent proof with a similar flavor to the proof of Lemma~\ref{lemma:contractible_H}, i.e., we will identify the set $\ca S_s$ as a continuous image of the solution set of an LMI.
  \begin{lemma}
    \label{lemma:connected_S}
    $\ca S_s$ is path-connected.
    \end{lemma}
    \begin{proof}
By Theorem $3$ in~\cite{de1999new}, $K \in \ca S$ if and only if there exists $X \succ 0$,
$G \in  \bb M_{m \times n} (\bb R)$ and $L \in \bb M_{n \times n} (\bb R)$ such that the following LMI is feasible
  \begin{align*}
    \begin{pmatrix}
      X & AG + BL \\
      G^\top A^\top + L^\top B^\top & G + G^\top - X
    \end{pmatrix} \succ 0.
  \end{align*}
  Note that $\ca S$ is the image of the continuous map $\psi: (X, L, G) \mapsto LG^{-1}$. But the solution set is convex (since if $(X_1, L_1, G_1)$ and $(X_2, L_2, G_2)$ are feasible, then $(1-\lambda)(X_1, L_1, G_1) + \lambda(X_2, L_2, G_2)$ is also feasible for the LMI when $\lambda \in (0,1)$). Hence $\ca S$ is connected.
  \end{proof}
  \begin{remark}
    One is tempted to use the LMI to get a diffeomorphism as Lemma~\ref{lemma:diffeo_H}. However, such a construction is not straightforward; the LMI devised in~\cite{de1999new} involves a linear inequality that leads to a non-injective $\psi$.
  \end{remark}
  \subsection{Connectedness of Structured Feedback Gains}
  If $K$ is constrained to a linear subspace $\ca U$ with $B=I$ (corresponding to the case where agents have direct control over their own dynamics; see \S\ref{sec:structured}), a sufficient condition for the set of stabilizing feedback gains to be connected is that $A \in \ca U$.
  \begin{lemma}
    \label{lemma:subspace_d}
    If $B=I$ and $A \in \ca U$, then the set $\{K \in \bb M_{n \times n}(\bb R): \rho(A-K)< 1, K \in \ca U\}$ is connected.
  \end{lemma}
  \begin{proof}
    By translation, without loss of generality, we suppose that $A=0$. We only need to observe that $K \in \ca K$, $t \mapsto (1-t)K + t 0$ is a continuous path between $K$ and $0$.
  \end{proof}
  In the case where only a subset of agents have direct control over their own dynamics, a sufficient condition to guarantee connectedness of $\ca K$ is that the structure of the corresponding system $A$ matches the graph topology and the entries of $A$ are nonnegative. Moreover, the closed-loop system $A-BK$ is constrained to be a nonnegative (feedback) system~\cite{luenberger1979introduction}.
  \begin{lemma}
    Suppose that $B=\begin{pmatrix}
      I_{m \times m}\\
      \bf 0_{(n-m) \times m}
    \end{pmatrix}$, $\ca U \subseteq \bb M_{m \times n}(\bb R)$ is a linear subspace, $A$ is nonnegative and the first $m$ rows $A_{1:m} \in \ca U$. The set $\{K \in \bb M_{m \times n}(\bb R): K \in \ca U, \rho(A-BK) < 1, A-BK \text{ is nonnegative}\}$ is connected.
  \end{lemma}
  \begin{proof}
    We may first choose $K_0$ such that the first $m$ rows of $A-BK_0$ is $0$. Hence without loss of generality, we assume that the first $m$ rows of $A$ are zero. If $K \in \ca K$, then $A-BK$ has the form,
\begin{align*}
A-BK =
  \begin{pmatrix}
    K\\
    A_{m:n, \cdot}
  \end{pmatrix}.
\end{align*}
We consider the convex path $t \mapsto (1-t)K + t0 \eqqcolon K_t$. But according to Gelfand's formula, putting $A_t = A-BK_t$ results in,
\begin{align*}
  \rho(A_t) = \lim_{k \to \infty}\|A_t^k\|_F^{\frac{1}{k}} \le \lim_{k \to \infty}\|A_0^k\|_F^{\frac{1}{k}} = \rho(A_0) < 1.
\end{align*}
Thereby, the convex path $t \mapsto (1-t)K + t0$ is contained in $\ca K$. As such, $\ca K$ is connected.
  \end{proof}
  In~\cite{feng2018on}, an instance was constructed to show that the set of stabilizing structured feedback gains for continuous systems could have exponentially many connected components. We shall provide an analogous construction for discrete-time systems.
  \begin{proposition}
    \label{prop:discrete_2}
    If $A = \begin{pmatrix} 0 & a^2 \\ 0 & 0 \end{pmatrix}$ with $|a| > 2$, $B=I$ and a subspace $\ca U = \{U \in \bb M_{2 \times 2}(\bb R): u_{ii} = 0 \text{ for }i=1,2\}, -a^2 u_{21}=u_{12}$, then the set $\ca K = \{K \in \bb M_{2 \times 2}(\bb R): K \in \ca U, \rho(A-K)< 1\}$ has two connected components.
  \end{proposition}
  \begin{proof}
    We observe that $\ca K$ can be parametrized by $1$ parameter $\alpha$, i.e., if $K \in \ca K$, then $K = \begin{pmatrix}
      0 & a^2\alpha \\
      -\alpha & 0
    \end{pmatrix}$. It follows that $A-K$ has the affine form,
    \begin{align*}
      \begin{pmatrix}
        0 & a^2(1-\alpha) \\
        \alpha & 0
      \end{pmatrix}.
    \end{align*}
    The modulus of the eigenvalues of this matrix is given by $\kappa = |a| \sqrt{|\alpha-\alpha^2|}$. In particular, if $|a| > 2$, the inequality $\kappa < 1$ holds over,
 \begin{align*}
   \alpha \in \left(\frac{|a| - \sqrt{a^2+4} }{2|a|}, \frac{|a|-\sqrt{a^2 - 4} } {2|a|}\right) \cup \left(\frac{|a| + \sqrt{a^2-4}}{2|a|}, \frac{|a|+\sqrt{a^2 + 4}}{2|a|}\right).
 \end{align*}
  \end{proof}
  
  We will now extend the above construction to show that over an instance of the pair $(A, B)$, the set $\ca K$ will have $2^{\floor{\frac{n}{2}}}$ connected components. We consider the case where $n$ is even in Lemma~\ref{lemma:exponential_components_discrete}. When $n$ is odd, we consider a block diagonal matrix by using $A$ with an even dimension of $n-1$ and a constant $0$ on its diagonal.
  \begin{lemma}
    \label{lemma:exponential_components_discrete}
    Suppose that 
    \begin{align*}
      A = \begin{pmatrix}
        0 & a_1^2 & & & & & & & \\
        0 & 0 & & & & & & &\\
        & & 0 & a_2^2 & & & && \\
        & & 0 & 0 & & & && \\
        & & & & \ddots & & && \\
        & & & & & \ddots &  & & \\
        & & & & & &  & 0&a_k^2 \\
        & & & & &  &  & 0&0 \\
      \end{pmatrix}, \qquad B = I,
    \end{align*} where $|a_j| >2$, for every $j=1, \dots, k$. Let $\ca U \subseteq \bb M_{n \times n}(\bb R)$ be a linear subspace defined by
    \begin{align*}
      \ca U = \{U \in \bb M_{n \times n}(\bb R): u_{12} = -a_1^2 u_{21}, \dots, u_{2k-1, 2k} = -a_k^2 u_{2k, 2k-1}\}.
    \end{align*}
    Then the set $\ca K =\{K \in \bb M_{n \times n}(\bb R): K \in \ca U, K \in \ca S_s\}$ has exactly $2^{k}$ connected components.
  \end{lemma}
  \begin{proof}
    By Proposition~\ref{prop:discrete_2}, each block has exactly two connected components. We only need to observe that $\ca K$ is the Cartesian product of the intervals defined in Proposition~\ref{prop:discrete_2}.
  \end{proof}
  \section{Conclusion}
  \label{sec:conclusion}
In this paper, we have provided topological and metrical insights into the set of 
stabilizing state feedback gains and MIMO output feedback gains for continuous and discrete time LTI systems for unstructured and structured synthesis.
The motivation for this work stems from recent interest in devising learning type algorithms for
control synthesis, which evolve over the set of stabilizing feedback gains. 
This in turn, has inspired the need to further examine the topological properties of these sets.
We envisage that some of these properties might been observed in the earlier literature in the system theory and
known to experts;
however, this work is an attempt to gather and prove these
properties in a concise and rigorous manner using basic topology. This paper is an extension of our work on SISO LTI (output feedback) systems that more heavily relies on the theory of polynomials for its anlaysis~\cite{bu2019siso}.

\section*{Acknowledgments}
 The authors acknowledge their discussions with Maryam Fazel, Sham Kakade, and Rong Ge, exploring connections between control theory and learning. This research was supported by DARPA Lagrange Grant FA8650-18-2-7836.

\bibliographystyle{siamplain}
\bibliography{ref}

\begin{thebibliography}{10}

\bibitem{brunovsky_K_1970}
{\sc P.~Brunovsk{\`y}}, {\em A classification of linear controllable systems},
  Kybernetika, 6 (1970), pp.~173--188.

\bibitem{bu2019siso}
{\sc J.~Bu, A.~Mesbahi, and M.~Mesbahi}, {\em On topological properties of the
  set of stabilizing feedback gains}.
\newblock under review, 2019.

\bibitem{Byrnes_book_1980}
{\sc C.~I. Byrnes}, {\em Algebraic and geometric aspects of the analysis of
  feedback systems}, in Geometrical Methods for the Theory of Linear Systems,
  C.~I. Byrnes and C.~F. Martin, eds., Springer Netherlands, 1980, pp.~85--124.

\bibitem{de1999new}
{\sc M.~C. de~Oliveira, J.~Bernussou, and J.~C. Geromel}, {\em A new
  discrete-time robust stability condition}, Systems \& control letters, 37
  (1999), pp.~261--265.

\bibitem{dullerud2013course}
{\sc G.~Dullerud and F.~Paganini}, {\em A Course in Robust Control Theory: A
  Convex Approach}, New York: Springer-Verlag, 2000.

\bibitem{Fam_TAC_1978}
{\sc A.~Fam and J.~Meditch}, {\em A canonical parameter space for linear
  systems design}, IEEE Transactions on Automatic Control, 23 (1978),
  pp.~454--458.

\bibitem{feng2018on}
{\sc H.~Feng and J.~Lavaei}, {\em On the exponential number of connected
  components for the feasible set of optimal decentralized control problems},
  in 2019 American Control Conference (ACC), 2019.

\bibitem{Hermann_TAC_1977no}
{\sc R.~Hermann and C.~Martin}, {\em Applications of algebraic geometry to
  systems theory--part i}, IEEE Transactions on Automatic Control, 22 (1977),
  pp.~19--25.

\bibitem{Hung_LAA_1998}
{\sc Y.~Hung and D.~Chu}, {\em Relationships between discrete-time and
  continuous-time algebraic {R}iccati inequalities}, Linear Algebra and its
  Applications, 270 (1998), pp.~287--313.

\bibitem{Jonckheere_TAC_1981}
{\sc E.~Jonckheere}, {\em On the existence of a negative semidefinite,
  antistabilizing solution to the discrete-time algebraic {R}iccati equation},
  IEEE Transactions on Automatic Control, 26 (1981), pp.~707--712.

\bibitem{kimura1975pole}
{\sc H.~Kimura}, {\em Pole assignment by gain output feedback}, IEEE
  Transactions on Automatic Control, 20 (1975), pp.~509--516.

\bibitem{lee2010introduction}
{\sc J.~Lee}, {\em Introduction to {T}opological {M}anifolds}, Springer Science
  \& Business Media, 2nd~ed., 2011.

\bibitem{luenberger1979introduction}
{\sc D.~Luenberger}, {\em Introduction to Dynamic Systems: Theory, Models, and
  Applications}, Wiley, 1979.

\bibitem{Mehrmann_LAA_1996}
{\sc V.~Mehrmann}, {\em A step toward a unified treatment of continuous and
  discrete time control problems}, Linear Algebra and its Applications, 241-243
  (1996), pp.~749 -- 779.

\bibitem{ober1987topology}
{\sc R.~J. Ober}, {\em Topology of the set of asymptotically stable minimal
  systems}, International Journal of Control, 46 (1987), pp.~263--280.

\bibitem{ohara1992differential}
{\sc A.~Ohara and S.-i. Amari}, {\em Differential geometric structures of
  stable state feedback systems with dual connections}, in System Structure and
  Control, 1992, pp.~176--179.

\bibitem{ohara1993geometric}
{\sc A.~Ohara and T.~Kitamori}, {\em Geometric structures of stable state
  feedback systems}, IEEE Transactions on Automatic Control, 38 (1993),
  pp.~1579--1583.

\bibitem{Prakash_IJC_1983}
{\sc M.~N. Prakash and A.~T. Fam}, {\em A geometric approach to stabilization
  by output feedback}, International Journal of Control, 37 (1983),
  pp.~111--125.

\bibitem{serre2002matrices}
{\sc D.~Serre}, {\em Matrices: Theory and applications}, 2010.

\bibitem{tu2010introduction}
{\sc L.~Tu}, {\em An Introduction to Manifolds}, Universitext, Springer New
  York, 2010.

\bibitem{Vidyasagar_TAC_1982}
{\sc M.~Vidyasagar, H.~Schneider, and B.~Francis}, {\em Algebraic and
  topological aspects of feedback stabilization}, IEEE Transactions on
  Automatic Control, 27 (1982), pp.~880--894.

\bibitem{wang1996grassmannian}
{\sc X.~A. Wang}, {\em Grassmannian, central projection, and output feedback
  pole assignment of linear systems}, IEEE Transactions on Automatic Control,
  41 (1996), pp.~786--794.

\bibitem{Wimmer_JMSEC_1995}
{\sc H.~K. Wimmer}, {\em On the existence of a least and negative-semidefinite
  solution of the discrete-time algebraic {R}iccati equation}, Journal of
  Mathematical Systems, Estimation, and Control, 5 (1995), pp.~445--457.

\bibitem{zabczyk2009mathematical}
{\sc J.~Zabczyk}, {\em Mathematical Control Theory: an Introduction}, Springer
  Science \& Business Media, 2008.

\end{thebibliography}

\end{document}